\documentclass[12pt]{article}
\usepackage[utf8]{inputenc}
\usepackage[english]{babel}
\usepackage{a4wide,cite}
\usepackage{graphicx}
\usepackage{color}
\usepackage{forloop}
\usepackage{amsmath}
\usepackage[version=3]{mhchem}
\usepackage{dsfont}
\usepackage{amssymb}
\usepackage{amsthm}
\usepackage{enumerate}
\usepackage{ulem}
\usepackage{changebar}

\usepackage{txfonts}
\usepackage{dsfont}
\usepackage{lmodern}
\usepackage[fleqn,tbtags]{mathtools}
\usepackage{array}
\usepackage{framed}
\usepackage{fancybox}
\usepackage{comment}
\usepackage{ifthen}
\usepackage{fancyhdr}
\usepackage{todonotes}

 {\endMakeFramed}

\newcounter{thm}
\numberwithin{equation}{section}
\numberwithin{thm}{section}

\newtheorem{lemma}[thm]{Lemma}
\newtheorem{proposition}[thm]{Proposition}

\newtheorem{theorem}[thm]{Theorem}
\theoremstyle{definition}

\newtheorem*{remark*}{Remark}
\newtheorem*{example*}{Example}

\newcommand{\linePage} {\noindent\makebox[\linewidth]{\rule{\textwidth}{1pt}} \\}
\specialcomment{simple} {$ $ \\ \begingroup   \scriptsize Simple Calculation \vspace{-4mm}   \\ \linePage   \endgroup \begingroup   }{$ $ \vspace{-3mm} \\ \linePage  \endgroup}

\excludecomment{simple}

\newcommand{\dd} {\mathrm{d}}

\newcommand{\tr} {\operatorname{tr}}

\newcommand{\braket}[2]{\langle #1 | #2 \rangle}

\newcommand{\norm}[1]{\| #1 \|}

\newcommand{\R}{\mathbb{R}}

\newcommand{\N} {\mathbb N}

\newcommand{\E} {\mathcal E}

\newcommand{\ol}[1]{\overline{#1}}

\newcommand{\ve}{\varepsilon}

\newcommand{\rk} {\right}
\newcommand{\lk} {\left}

\newcommand {\numberthis} {\addtocounter{equation}{1}\tag{\theequation}}

\newcommand{\vn}{{\vec n}}
\newcommand{\Egl}{{\mathcal{E}_g^\ell}}

\title{Triviality of a model of particles with point interactions in  the thermodynamic limit\footnote{\copyright\ 2016 by the authors. This work may be
   reproduced, in its entirety, for non-commercial purposes.}}
\author{Thomas Moser and Robert Seiringer
\\ \\  \small{Institute of Science and Technology Austria (IST Austria)}\\  \small{Am Campus 1, 3400
  Klosterneuburg, Austria}}

\date{November 17, 2016}

\begin{document}

\maketitle

\begin{abstract}
We consider a  model of fermions interacting via point interactions, defined via a certain weighted Dirichlet form. While for two particles the interaction corresponds to infinite scattering length, the presence of further particles effectively decreases the interaction strength. We show that the model becomes trivial in the thermodynamic limit, in the sense that the free energy density at any given particle density and temperature agrees with the corresponding expression for non-interacting particles. 
\end{abstract}

\smallskip
\noindent \textbf{Mathematics Subject Classification (2010).} 81Q10, 46N50, 82B10

\noindent \textbf{Keywords.} point interactions, scattering length, Fermi gas, thermodynamic limit


\section{Introduction}

Due to their relevance for cold-atom physics \cite{zwerger}, quantum-mechanical models of particles with zero-range interactions have recently received a lot of attention. Of particular interest is the unitary limit of infinite scattering length, where one has scale invariance due to the lack of any intrinsic length scale (see, e.g., \cite{braaten,burovski,goulko,jonsell,werner}). Despite some effort \cite{dell,finco,correggi,correggi2,MS}, it remains an open problem to establish the existence of a many-particle model with two-body point interactions. 
Such a model is known to be unstable in the case of bosons (a fact known as Thomas effect \cite{braaten,Thomas1935,correggi}, closely related to the Efimov effect \cite{efimov,Tamura1991,Yafaev74}) and hence can only exist for  fermionic particles. 
In contrast, the two-body problem is completely understood and point interactions can be characterized via self-adjoint extensions of the Laplacian on $\R^3\setminus \{0\}$ (see \cite{albeverio} for details). These self-adjoint extensions can be interpreted as corresponding to an attractive point interaction, parametrized by the scattering length $a$, with interaction strength increasing with $1/a$. For  non-positive scattering length, $a \leq 0$, the attraction is too weak to support bound states, while there exists a negative energy bound state for $a>0$.

In the case of non-positive scattering length, $a\leq 0$, corresponding to the absence of two-body bound states, point interactions can alternatively be defined via the quadratic form
\begin{equation}\label{qf2}
\int_{\R^3} \left( \frac 1{|x|} - \frac 1 a\right)^2 \left| \nabla f(x)\right|^2 \, \dd x \qquad \text{on}\quad L^2(\R^3, (|x|^{-1}-a^{-1})^2 \dd x)
\end{equation}
The unitary limit corresponds to $a^{-1} = 0$. Recall that the scattering length is defined (see, e.g., \cite[Appendix~C]{LSSY}) via the asymptotic behavior of the solution to the zero-energy scattering equation, which in this case is simply equal to $|x|^{-1} - a^{-1}$, corresponding to $f\equiv 1$. To see that \eqref{qf2} corresponds to a point interaction at the origin, note that an integration by parts shows that 
\begin{align}\nonumber
\int_{|x|\geq \epsilon} \left( \frac 1{|x|} - \frac 1 a\right)^2 \left| \nabla f(x)\right|^2 \, \dd x & = \int_{|x|\geq \epsilon}  \left| \nabla \left( \frac 1{|x|} - \frac 1 a\right) f(x)\right|^2 \, \dd x \\ & \quad - \int_{|x|=\epsilon} \left( \frac 1{|x|} - \frac 1 a\right) \frac 1{|x|^2} |f(x)|^2 \dd\omega  \label{1.2}
\end{align} 
for any $\epsilon>0$. The last term vanishes as $\epsilon\to 0$ if $f$ vanishes faster than ${|x|^{1/2}}$ at the origin. 

We consider here a many-body generalization of \eqref{qf2}, which was introduced in \cite{albeverio2}. It has the advantage of being manifestly well-defined, via a non-negative Dirichlet form. As already noted above, in general it is notoriously hard to define many-body systems with point interactions, see \cite{correggi,correggi2,dell,finco,MS}, due to the inherent instability problems.  The model under consideration here was studied in \cite{Frank} were it was shown to satisfy a Lieb--Thirring inequality, i.e., the energy can be bounded from below by a semiclassical expression of the form $C \int \rho(x)^{5/3}\dd x$, with $\rho$ the particle density and $C$ a positive constant. Up to the value of $C$, this is the same as the inequality for non-interacting fermions used by Lieb and Thirring \cite{LT1,LT2} in their proof of stability of matter. (For other recent work on Lieb-Thirring inequalities for interacting particles, see \cite{Lu1,Lu2,Lu3,Lu4}.)

The model considered here has the disadvantage that the interaction is not purely two-body, however. In fact, it is a full many-body interaction, its strength depends on the position of all the particles and is weakened due to their presence. We shall show here that the effects of the interaction actually disappear in the thermodynamic limit, and the thermodynamic free energy density agrees with the one for non-interacting fermions.

In the next section, we shall introduce the model and explain our main results. The rest of the paper is devoted to their proof.

\section{Model and main results}

For $N\geq 2$, $\vec x =(x_1,\dots,x_N) \in \R^{3N}$, let $g : \R^{3N}\to \R$ denote the function 
\begin{gather*}
g(\vec x) = \sum_{1\leq i < j \leq N } \frac 1 {|x_i - x_j|}\,. \numberthis
\label{eqG}
\end{gather*}
We consider fermions with $q\geq 1$ internal (spin) states, described by wave functions in the  subspace $\mathcal{A}_q^N \subset L^2( (\R^3 \times \{1,\dots,q\})^N,g(\vec x)^2 \dd \vec x)$ of functions that are  totally antisymmetric with respect to permutations of the variables $y_i = (x_i, \sigma_i)$, where $x_i\in\R^3$ and $\sigma_i\in \{1,\dots, q\}$.  For $\psi\in \mathcal{A}_q^N$, our model is defined via the quadratic form
\begin{equation}\label{defeg}
\E_g(\psi) =  \sum_{i=1}^N \int_{\R^{3N}} g(\vec x)^2 |\nabla_i \psi(\vec y)|^2 \dd \vec y
\end{equation}
where $\nabla_i$ stands for the gradient with respect to $x_i \in \R^3$, and we introduced the shorthand notation $\int \dots \dd \vec y = \sum_{\vec \sigma}  \int \dots \dd \vec x$ with $\vec \sigma = (\sigma_1,\dots,\sigma_N)$. Since $g$ is a harmonic function away from the planes $\{x_i=x_j\}$ of particle intersection, an integration by parts as in \eqref{1.2} shows that \eqref{defeg} corresponds to a model of point interactions, as $\E_g(\psi) = \sum_{i=1}^N \int |\nabla_i  g\psi|^2 $ in case $\psi$ has compact support away from these planes. More generally, $\E_g(\psi) = \sum_{i=1}^N \int |\nabla_i  g\psi|^2 $ holds if $\psi$ vanishes faster than the square root of the distance to the planes of intersection, which is in particular the case for smooth and completely antisymmetric functions of the spatial variables. In other words, the model is trivial for $q=1$.

For $N$ particles in a cubic box $[0,L]^3 \subset \R^3$, the free energy at temperature $T = \beta^{-1} >0$ is defined as usual as 
\begin{equation}
F_g = -T \ln \tr e^{-\beta H_g} 
\end{equation}
where $H_g$ denotes the operator defined by the quadratic form \eqref{defeg},  restricted to functions in $\mathcal{A}_q^N\cap H^1(\R^{3N};g(\vec x)^2 \dd \vec x)$ with support in $([0,L]^3)^{ N}$. The latter restriction corresponds to choosing Dirichlet boundary conditions on the boundary of the cube $[0,L]^3$. 
Alternatively, one can use the variational principle \cite[Lemma~14.1]{StabMatt} to write the free energy as
\begin{gather*}
F_{g}(\beta,N,L) = - T \ln \sup_{\substack{\{\psi_k\} \\ \braket{\psi_i}{\psi_j}_g = \delta_{ij} }} \sum_k  e^{- \beta \E_g(\psi_k) } \label{eqFreeEnergy} \numberthis
\end{gather*}
where $\langle\,\cdot\,|\,\cdot\,\rangle_g$ denotes the inner product on $L^2( (\R^3\times \{1,\dots,q\})^N,g(\vec x)^2 \dd\vec x)$, 
\begin{gather*}
\braket {\psi_i}{\psi_j}_g =  \int_{\R^{3N}} g^2(\vec x) \ol {\psi_i(\vec y)} \psi_j(\vec y) \dd \vec y, \numberthis
\label{eqScalarProduct}
\end{gather*}
and the supremum is over all finite sets of orthonormal functions in $\mathcal{A}_q^N$ with support in $([0,L]^3)^{ N}$.  We are interested in the thermodynamic limit
\begin{equation}
f_g(\beta ,\rho) = \lim_{N\to \infty}  \frac{ \rho}{N} F_g(\beta,N,(N/\rho)^{1/3})
\end{equation}
where $\rho>0$ denotes the particle density. 

In the non-interacting case corresponding to taking $g\equiv 1$, the free energy density can be evaluated explicitly, and is given by \cite{thirring}
\begin{equation}\label{fni}
f(\beta,\rho) = \sup_{\mu\in \R} \left [ \mu \rho - \frac {qT } {(2\pi)^3} \int_{\R^3} \ln \left ( 1 + e^{- \beta (p^2 -\mu) }\right)  dp\right]
\end{equation}
Our main result shows that the two functions, $f_g$ and $f$, are actually identical. 

\begin{theorem}
For any $\beta>0$ and $\rho>0$, and any $q\geq 1$, 
\begin{gather*}
f_g(\beta,\rho) = f(\beta, \rho)  \numberthis \label{eqthm}
\end{gather*}
\label{thmAll}
\end{theorem}

We shall actually prove a stronger result below, namely a lower bound on $F_g(\beta,N,L)$ for finite $N$ which agrees with the corresponding expression for non-interacting particles, $F(\beta,N,L)$, to leading order in $N$, with explicit bounds on the correction term. Note that the corresponding upper bound is trivial, since for functions $\phi \in C_0^\infty((\R^{3}\times\{1,\dots,q\})^N)$
\begin{equation}
\E_g( \phi/ g) = \sum_{i=1}^N \int |\nabla_i  \phi(\vec y)|^2 \, \dd \vec y
\end{equation}
and hence $F_g(\beta,N,L)\leq F(\beta,N,L)$. Moreover, as already noted above one has  $F_g(\beta,N,L)= F(\beta,N,L)$ for $q=1$, since functions in $\mathcal{A}_1^N$ vanish whenever $x_i=x_j$ for some $i\neq j$. Hence it suffices to consider the case $q\geq 2$.

Theorem~\ref{thmAll} also holds true for the ground state energy, i.e., $\beta=\infty$, where $f(\infty,\rho) = \frac 35 (6\pi^2/q)^{2/3} \rho^{5/3}$. The proof of the equality \eqref{eqthm} in  this case is actually substantially easier, as the analysis of the entropy in Section~\ref{sec:entropy} is not needed.

Intuitively, the result in Theorem~\ref{thmAll} can be explained via a comparison of \eqref{defeg} with \eqref{qf2}. Effectively, the scattering process between two particles, $i$ and $j$, say, corresponds to a non-zero scattering length of the form
\begin{equation}
-\frac 1{a_{\rm eff}} = \sum_{\{k,l\}\neq \{i,j\}} \frac 1{|x_k-x_l|}\,.
\end{equation}
In the limit of large particle number, the sum of these other terms diverges, corresponding to an effective scattering length zero, i.e., no interactions. 

A minor modification of the proof shows that Theorem~\ref{thmAll} also holds for a model where the function $1/|x|$ in \eqref{eqG} is replaced by $1/|x| - 1/a$ for $a\leq 0$, corresponding to a two-body interaction with negative scattering length $a$. This only increases the effective scattering length $a_{\rm eff}$.

From Theorem~\ref{thmAll} we conclude that the model \eqref{defeg} is not suitable to describe a gas of fermions with point interactions, as it becomes trivial in the thermodynamic limit. No non-trivial models that are proven to be stable for arbitrary particle number exist to this date, however. Such non-trivial models are not expected to be given by a Dirichlet form of the type \eqref{defeg}, since such forms are naturally well-defined even in the bosonic case, where point-interaction models are known to become unstable due to the Thomas effect \cite{braaten,Tamura1991,Yafaev74,Thomas1935,correggi,efimov}.

In the remainder of this paper, we shall give the proof of Theorem~\ref{thmAll}. We start with a short outline of the main steps in the next section.

\section{Outline of the proof}

In the first step in Section~\ref{secBoundaryLocalized} we shall localize particles in small boxes. This part of the  Dirichlet--Neumann bracketing technique is quite standard, but it does not directly allow us to reduce the problem to fewer particles, as the interactions depend on the location of all the particles, including the ones in different boxes. Still this step allows us to compare our model with the corresponding one for non-interacting fermions, by utilizing a suitable version of the Hardy inequality to quantify the effect of the deviation of the weight function $g$ in \eqref{eqG} from being a constant. This analysis is done in Section~\ref{sec:enno}. Note that the relevant constant to compare $g$ with depends on the distribution of the particles in the various boxes, hence the importance of the first step. An important point in the analysis is a control on the particle number distribution, which is obtained in Prop.~\ref{prop:nj}. 

In Section~\ref{sec:entropy} we shall give a rough bound on the entropy for large energy, which will allow us to conclude that that to compute the free energy \eqref{eqFreeEnergy}, it suffices to consider only states with energy $E\lesssim N \ln N$. We do this by applying the localization technique to very small boxes, with side length decreasing with  energy, in order to have to consider  effectively only the ground states in each small box.

In the low energy sector, corresponding to energies $E\lesssim N \ln N$, our bounds in Section~\ref{sec:enno} allow to make a direct comparison of our model with non-interacting fermions. This comparison is detailed in Section~\ref{sec:low}. For this purpose, we shall choose much larger boxes than in the previous step, very slowly increasing to infinity with $N$ in order for finite size effects to vanish in the thermodynamic limit. Finally, Section~\ref{secFreeEnergyConvergence} collects all the results in the previous sections to give the proof of Theorem~\ref{thmAll}.

Throughout the proof, we shall use the letter $c$ for universal constants independent of all parameters, even though $c$ might have different values at different occurrences. Similarly, we use $c_\eta$ for functions of 
$
\eta = \beta\rho^{2/3} 
$
 that are uniformly bounded for $\eta>\ve$ for any $\ve>0$. Note that the free energy for noninteracting particles in \eqref{fni} satisfies the scaling relation
 \begin{equation}
f(\beta,\rho)=\rho^{5/3} f(\eta,1) \ , \qquad \eta = \beta\rho^{2/3} \,,
\end{equation}
and $\eta\to \infty$ corresponds to the zero-temperature limit.

\section{Particle localization in small boxes}
\label{secBoundaryLocalized}

Given an integer $m\geq 2$, we shall divide the cube $[0,L]^3$ into $M= m^3$ disjoint cubes of side length $\ell = L/m$, denoted by  $\{B_i\}_{i=1}^{M}$. In order to obtain a lower bound on $\E_g$, we  
 introduce Neumann boundary conditions on the boundary of each box $B_i$. 
 
 Specifically, given a vector $\vn = \{ n_1,\dots ,n_M\}$ of nonnegative integers with $\sum_{j=1}^M n_j = N$, let $\mathcal{B}_{\rm sym}(\vn) $ denote the subset of $[0,L]^{3N}$ where exactly $n_j$ particles are in $B_j$, for all $1\leq j \leq M$. More precisely, if 
 \begin{equation}\label{def:tib}
{\mathcal{B}}(\vn) = B_1^{n_1} \times \cdots \times B_M^{n_M}
\end{equation}
and, for general $A \subset \R^{3N}$ and $\pi \in S^N$ (the permutation group of $N$ elements)
\begin{equation}
\pi(A) = \{ \vec x \, : \, \pi^{-1}(\vec x) \in A\} \ , \qquad \pi(\vec x) : = (x_{\pi(1)},\dots, x_{\pi(N)}) 
\end{equation}
we have 
\begin{equation}
\mathcal{B}_{\rm sym}(\vn)  = \bigcup_{\pi\in S^N} \pi( {\mathcal{B}}(\vn) )
\end{equation}
 Then clearly 
 \begin{equation}
1 = \sum_{\vn} \chi_{\mathcal{B}_{\rm sym}(\vn)}(\vec x)
\end{equation}
for almost every $\vec x \in [0,L]^{3N}$. 
Correspondingly one can write for any $\psi \in \mathcal{A}_q^N$  supported in $[0,L]^{3N}$
\begin{gather*}
\psi(\vec y) = \sum_{\vn} \chi_{\mathcal{B}_{\rm sym}(\vn)}(\vec x)   \psi(\vec y)  \eqqcolon \sum_{\vn}   \psi^{\vn}(\vec y)\,.\numberthis \label{rso}
\end{gather*}
Note that each $\psi^{\vn}$ is a function in $\mathcal{A}_q^N$ with the property that it is non-zero only if exactly $n_j$ particles are in $B_j$ for any $1\leq j\leq M$. In particular, the functions appearing in the decomposition on the right side of \eqref{rso} all have disjoint support.

Conversely, given a set of functions $\psi^{\vn} \in \mathcal{A}_q^N$ supported in $\mathcal{B}_{\rm sym}(\vn)$, we can define $\psi\in\mathcal{A}_q^N$ via \eqref{rso}. Hence there is a one-to-one correspondence between functions in $\mathcal{A}_q^N$ and sets of functions $\psi^{\vn}$. We now redefine our energy functional $\E_g$ as 
\begin{equation}\label{redef}
\Egl(\psi) =  \sum_{\vn}   \sum_{i=1}^N \int_{\mathcal{B}_{\rm sym}(\vn)} g(\vec x)^2 | \nabla_i \psi^{\vn}(\vec y)|^2 \dd \vec y 
\end{equation}
This coincides with the definition \eqref{defeg} in case $\psi \in H^1((\R^{3}\times\{1,\dots,q\})^N, g(\vec x)^2 \dd \vec x)$, but is more general since it allows for wave functions that are discontinuous at the boundaries of the $B_j$, effectively introducing Neumann boundary conditions there.

Note that with the definition \eqref{redef} above, we have 
\begin{gather*}
\Egl(\psi) = \sum_{\vn} \Egl(\psi^{\vn })  \quad \text {for $\psi = \sum_{\vn} \psi^{\vn}$} \label{eqSplit1} \numberthis
\end{gather*}
In particular, the corresponding operator is diagonal with respect to the direct sum decomposition of $\mathcal{A}_q^N$ into functions supported on $\mathcal{B}_{\rm sym}(\vn)$, and hence the min-max principle implies the bound
\begin{equation}\label{48}
\sup_{\substack{\{\psi_k\} \\ \braket{\psi_i}{\psi_j}_g = \delta_{ij} }} \sum_k  e^{- \beta \E_g(\psi_k) }  \leq \sum_{\vn} \sup_{\substack{\{\psi^{\vn}_k\} \\ \braket{\psi^{\vn}_i}{\psi^{\vn}_j}_g = \delta_{ij} }} \sum_k  e^{- \beta \Egl(\psi^{\vn}_k) } 
\end{equation}
where on the right side it is understood that each $\psi_j^{\vn}$ is supported in $\mathcal{B}_{\rm sym}(\vn)$.

As a final step in this section we want to simplify the problem by getting rid of the antisymmetry requirement for particles localized in different boxes. There exists a simple isometry between functions $\psi^{\vn}$ in $\mathcal{A}_q^{N}$ and functions whose support is on the smaller set ${\mathcal{B}}(\vn) $ in \eqref{def:tib}, 
where $x_1,\dots, x_{n_1} \in B_1$, $x_{n_1+1},\dots,x_{n_1+n_2} \in B_2$, etc., and which are antisymmetric only with respect to permutations of the $y_i$ corresponding to $x_i$ in the same box. This isometry is simply
\begin{equation}
\psi^{\vn} \mapsto     \left( \frac{ N!}{ \prod_{j=1}^M n_j! } \right)^{1/2}      \chi_{{\mathcal{B}}(\vn)} \psi^{\vn}
\end{equation}
Note that the normalization factor is chosen such that both sides have the same norm, and the left side can be obtained from the right by a suitable antisymmetrization over all variables $y_i$. Moreover, both functions yield the same value when plugged into $\Egl$. Let $\mathcal{A}_q^{N,\ell}(\vn)$ denote the set $\{ \chi_{\mathcal{B}(\vn)} \psi \, : \, \psi \in \mathcal{A}_q^N\}$, i.e., functions supported in ${\mathcal{B}}(\vn)$ that are antisymmetric in the variables corresponding to the same box. 
The bound \eqref{48} and the above observation imply that 
\begin{equation}\label{ja48}
F_{g}(\beta,N,L) \geq  - T \ln   \sum_{\vn} \sup_{\substack{\{\psi_k \in \mathcal{A}_q^{N,\ell}(\vn) \} \\ \braket{\psi_i}{\psi_j}_g = \delta_{ij} }} \sum_k  e^{- \beta \Egl(\psi_k) } 
\end{equation}

\section{Energy and norm bounds}\label{sec:enno} 

Our goal in this next step to derive a lower bound on $\Egl( \psi)$ for $\psi\in \mathcal{A}_q^{N,\ell}(\vn)$, i.e., functions  supported in ${\mathcal{B}}(\vn)$, and to compare the norm of such a $\psi$ with the standard, unweighted $L^2$ norm. For this purpose, we shall need a certain version of the Hardy inequality, which will be derived in the next subsection.

\subsection{Hardy inequalities}
\label{secHardy}

Recall the usual Hardy inequality 
\begin{gather}
\int_{\R^3} |\nabla f(x)|^2 \dd x \ge \frac {1} 4 \int_{\R^3} \frac{|f(x)|^2}{|x|^2} \dd x
\label{eqHardy}
\end{gather}
for functions $f\in \dot H^1(\R^3)$. We shall need a local version of (\ref{eqHardy}) on balls.

\begin{lemma}
\label{thmHardyBall}
Let $B_\ell \subset \R^3$ denote the open centered ball with radius $\ell$. For any $f \in H^1(B_\ell)$ 
\begin{gather}
2 \int_{B_\ell} |\nabla f(x)|^2 \dd x + \frac 9{2 \ell^2} \int_{B_\ell} |f(x)|^2 \dd x\ge  \frac 14 \int_{B_\ell} \frac {|f(x)|^2}{|x|^2} \dd x
  \label{eqHardyBall}
\end{gather}
\end{lemma}

\begin{proof}
We apply the Hardy inequality \eqref{eqHardy} to the function $h(x) = f(x) [1 - |x|/\ell]_+$, where $[\,\cdot\,]_+$ denotes the positive part. For the right side of \eqref{eqHardy} we obtain
\begin{gather*}
\frac 1 4 \int_{B_\ell} \frac {|h(x)|^2}{|x|^2} \dd x = \frac{1} 4 \int_{B_\ell} \frac {|f(x)|^2}{|x|^2}\left( 1 - \frac{2|x|}{\ell} + \frac{|x|^2}{\ell^2}\right) \dd x \\
\ge \frac{1-\ve}4  \int_{B_\ell} \frac {|f(x)|^2}{|x|^2} \dd x  - \frac {1-\ve}{4 \ve \ell^2}   \int_{B_\ell} |f(x)|^2 \dd x \numberthis
\end{gather*}
for any $\ve>0$. For the left side of \eqref{eqHardy} a simple Schwarz inequality yields
\begin{equation}
\int_{B_\ell} |\nabla h(x)|^2 \dd x \le (1+\delta)  \int_{B_\ell} |\nabla f(x)|^2 \dd x + \frac {1+\delta}{\delta \ell^2} \int_{B_\ell} |f(x)|^2 \dd x  
\end{equation}
for $\delta > 0$. 
In combination we obtain the desired inequality \eqref{eqHardyBall} by choosing $\ve=1/6$ and $\delta=2/3$.
\end{proof}

For later use we need a version of Lemma \ref{thmHardyBall} on cubes with arbitrary location relative to the singularity. 

\begin{lemma}
\label{thmHardyBallBox}
Let $C_\ell = [0,\ell]^3$. For any $y\in\R^3$ and  any $f \in H^1(C_\ell)$,
\begin{gather}
c_0 \int_{C_\ell} |\nabla f(x)|^2 \dd x + \frac{c_1 }{ \ell^{2}} \int_{C_\ell} |f(x)|^2 \dd x\ge \frac 14\int_{C_\ell} \frac {|f(x)|^2}{|x-y|^2} \dd x
\label{eqHardyBallBox}
\end{gather}
with $c_0\leq 16$ and $c_1\leq 144$.
\end{lemma}

The stated bounds on the constants $c_0$ and $c_1$ are presumably far from optimal, but suffice for our purpose.

\begin{proof} 
If $y\not\in C_\ell$, we can replace it by the point in $C_\ell$ closest to $y$. This can only increase the right side. Hence we may assume that 
$y\in C_\ell$. Let $B$ denote the ball of radius $\ell/2$ around $y$. Then
\begin{equation}\label{neweq}
\frac 14 \int_{C_\ell\setminus B} \frac {|f(x)|^2}{|x-y|^2} \dd x \leq \frac 1{\ell^2} \int_{C_\ell\setminus B}  |f(x)|^2 \,\dd x
\end{equation}
Define a function $\tilde f$ by extending $f$ to $[-\ell,2\ell]^3$ as 
\begin{gather}
\tilde f(x_1,x_2,x_3) = f(\tau(x_1), \tau(x_2), \tau(x_3))
\end{gather}
where
\begin{gather}
 \tau(x) \coloneqq \begin{cases} -x & x \in [-\ell,0] \\ x & x \in[0,\ell] \\ 2\ell-x & x \in[\ell,2\ell]
\end{cases}
\label{eqTau}
\end{gather}
Then $\tilde f \in H^1([-\ell,2\ell]^3)$. Since  $B\subset [-\ell,2\ell]^3$, we get with the aid of 
 the Hardy inequality \eqref{eqHardyBall} on $B$ (with $\ell/2$ in place of $\ell$)
\begin{align*}
&\frac 14 \int_{C_\ell \cap B} \frac {|f(x)|^2}{|x-y|^2} \dd x  \le \frac 14 \int_{B} \frac {|\tilde f(x)|^2}{|x-y|^2}\dd x \\ & \le 2 \int_{B} |\nabla \tilde f(x)|^2 \dd x +\frac {18}{ \ell^{2}} \int_{B} |\tilde f(x)|^2 \dd x \\ 
&\le 8 \lk ( 2 \int_{C_\ell \cap B} |\nabla  f(x)|^2 \dd x + \frac{18}{ \ell^{2}} \int_{C_\ell \cap B} | f(x)|^2 \dd x  \rk )  \numberthis
\label{eqProofBallBox}
\end{align*}
In the last step, we used that $B$ intersects, besides $C_\ell$, at most $7$ other translates of $C_\ell$, and that the intersection of $B$ with these translates are, when reflected back to $C_\ell$, contained in $C_\ell \cap B$ (see Fig.~\ref{figBB1}). In combination, \eqref{neweq} and \eqref{eqProofBallBox} imply \eqref{eqHardyBallBox}.
\end{proof}

\begin{figure}[htp]
\begin{center}
\includegraphics[width=8cm]{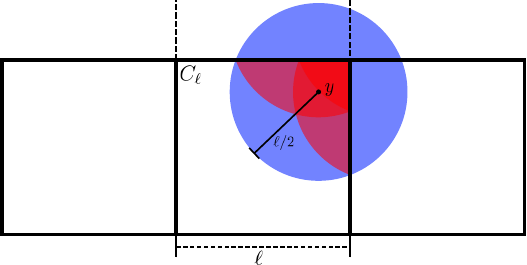}
  \caption{Two-dimensional illustration of the reflection  technique used in the proof of Lemma \ref{thmHardyBallBox}. The box $C_\ell$ and two  of its neighbor boxes are shown, as well as the ball $B$ around $y\in C_\ell$.  Using the extended function $\tilde f$ we can mirror  $C_\ell \setminus B$ back into  $C_\ell \cap B$. There are at most $8$ reflected components in three dimensions, the worst case being if the ball $B$ intersects with a corner of $C_\ell$.}
  \label{figBB1}
\end{center}
\end{figure}

\subsection{A lower bound on $\Egl$}

Let $\psi$ be an $L^2((\R^3\times \{1,\dots, q\})^N, g(\vec x)^2 \dd \vec y)$-normalized function in $\mathcal{A}_q^{N,\ell}(\vn)$, 
 defined just above  \eqref{ja48}.  Let $d_{jk}$ denote the distance between boxes $B_j$ and $B_k$. 
For $\vec x \in \mathcal{B}(\vn)$, we can bound 
\begin{equation}\label{lbg}
g(\vec x) \geq \sum_{1\leq j<k\leq M} \frac {n_j n_k} { d_{jk} + 2 \sqrt{3} \ell} + \sum_{j=1}^M \frac {n_j (n_j-1)} { 2\sqrt{3} \ell} \geq K_- + \frac V {4\sqrt{3} \ell}
\end{equation}
where
\begin{equation}\label{def:kmv}
K_- = \sum_{\substack{1\leq j<k\leq M\\d_{jk}>0}} \frac {n_j n_k} { d_{jk} + 2 \sqrt{3} \ell} \quad \text{and} \quad V = \sum_{j=1}^M n_j (n_j +m_j - 1)
\end{equation}
Here $m_k$ denotes the total number of particles in the $26$ neighboring boxes of $B_k$.
The bound \eqref{lbg} immediately leads to the lower bound
\begin{equation}\label{lbegl}
\Egl(\psi) \geq \left( K_- + \frac V {4\sqrt{3} \ell} \right)^{2}  \E^\ell(\psi)
\end{equation}
for $\psi\in \mathcal{A}_q^{N,\ell}(\vn)$, 
where $\E^\ell$ on the right side stands for the energy functional for noninteracting particles, corresponding to $g\equiv 1$ in \eqref{redef}.

\subsection{Bounds on norms}\label{ss:norms}

In the following, it will be necessary to compare the norm $\| \, \cdot \, \|_g = \langle \, \cdot\, | \,\cdot\, \rangle_g^{1/2}$ with the standard $L^2$ norm $\|\,\cdot\,\|$ without weight. For  $\psi\in \mathcal{A}_q^{N,\ell}(\vn)$,  the bound \eqref{lbg} immediately implies the lower bound
 \begin{equation}
\| \psi\|_g  \geq \left( K_- + \frac V {4\sqrt{3} \ell} \right)  \| \psi\| 
\end{equation}
To obtain a corresponding upper bound, we proceed as follows. 
For given $i$, corresponding to $x_i \in B_k$ for some box $B_k$, let $\mathcal{N}[i]$ be the set of $j$s  with $j\neq i$ such that $x_j$ is either in the same box $B_k$ or in one of the $26$ neighboring boxes touching $B_k$. With $m_k$ as defined above, $|\mathcal{N}[i]|= n_k+m_k-1$ for $x_i\in B_k$. 
Then
\begin{equation}\label{ubg}
g(\vec x) \leq    \frac 12 \sum_{i=1}^N   \sum_{ j \in \mathcal{N}[i] }  \frac 1 {|x_i-x_j|}  + K_+ \quad \text{with} \quad K_+ = \sum_{\substack{1\leq j<k\leq M\\d_{jk}>0}} \frac {n_j n_k} { d_{jk} }  
\end{equation}
for $\vec x \in \mathcal{B}(\vn)$. 
The Cauchy-Schwarz inequality  implies
\begin{equation}
\|\psi\|_g^2 \leq (1+\ve) K_+^2 \| \psi\|^2  +   \left( 1 + \ve^{-1}\right)  \frac V4 \sum_{i=1}^N  \sum_{ j \in \mathcal{N}[i] }   \int \frac{|\psi(\vec y)|^2}{|x_i-x_j|^2} \dd\vec y   
\end{equation}
for any $\ve>0$, where $V$ is defined in \eqref{def:kmv}. 
In the last term, we use the Hardy inequality \eqref{eqHardyBallBox} for the integration over $x_i$, and obtain
\begin{align}\nonumber
\|\psi\|_g^2 & \leq \left[(1+\ve) K_+^2 + \frac {c_1}{\ell^2} \left( 1 + \ve^{-1} \right) V^2  \right] \| \psi\|^2    \\ & \quad +   \left( 1 + \ve^{-1}\right)   c_0  V \sum_{i=1}^N  \left|\mathcal{N}[i] \right| \int |\nabla_i \psi(\vec y)|^2\, \dd\vec y   
\end{align}
If we reinsert $g(\vec x)^2$ into the last integrand using \eqref{lbg}, we thus obtain the following lemma.

\begin{lemma}\label{lemma53}
For $\psi \in \mathcal{A}_q^{N,\ell}(\vn)$, 
we have the bounds
\begin{align}\nonumber 
\left(  K_- + \frac V {4\sqrt{3} \ell}\right)^2 \|\psi\|^2 \leq \|\psi\|_g^2 & \leq (1+\ve)    \left[ K_+^2 + \frac {c_1}{\ve \ell^2} V^2  \right] \| \psi\|^2  \\ & \quad  +      \frac{(1+\ve^{-1} )c_0 V}{ \left(  K_- + \frac V {4\sqrt{3} \ell}\right)^2}   \sum_{i=1}^N  \left|\mathcal{N}[i] \right| \int |\nabla_i \psi(\vec y)|^2 g(\vec x)^2\, \dd\vec y    \label{lbng}
\end{align}
for any $\ve>0$, 
where $K_\pm$ and $V$ are defined in \eqref{def:kmv} and \eqref{ubg}, respectively.
\end{lemma}

\subsection{A bound on the number of particles in a box}

Let again $\psi$ be a wavefunction  in $\mathcal{A}_q^{N,\ell}(\vn)$ and let us assume it is normalized, i.e., $\|\psi\|_g = 1$.  We have the following  a priori bound. 

\begin{proposition}\label{prop:nj}
There exists a constant $\kappa>0$ 
 such that for any normalized $\psi\in \mathcal{A}_q^{N,\ell}(\vn)$ and any $\ell>0$ we have 
\begin{equation}\label{th}
\Egl(\psi) \geq \frac {\kappa}{q^{2/3} } \sum_{j=1}^M \frac { \left[ n_j - q\right]_+^{5/3}}{\ell^2}
\end{equation}
\end{proposition}

Here $[\,\cdot\,]_+ = \max\{0,\,\cdot\,\}$ denotes the positive part. The  bound \eqref{th} allows us to conclude that for all normalized $\psi\in \mathcal{A}_q^{N,\ell}(\vn)$ with $\Egl(\psi)< E$ we have $n_j\leq q$ for all $j$ if we choose $\ell$ such that $E\ell^2 q^{2/3}\leq \kappa$. Furthermore, for large $E\ell^2$ we get the bound $\max_j n_j \lesssim  q^{2/5} (E\ell^2)^{3/5}$. 

\begin{proof}
We use Lemma 3 from \cite{Frank} which states that  for a subset $A \subseteq \{1,\ldots,N\}$ corresponding to particles $x_k \in B_j$ for $k\in A$,
\begin{gather}
 \sum_{i \in A}   \int_{B_j^{|A|}} g(\vec x)^2 |\nabla_i \psi (\vec y)|^2 \dd \vec y_A \ge \frac {\tilde \kappa} {\ell^2} \left[ |A| - q \right]_+ \int_{B_j^{|A|}} g(\vec x)^2 |\psi(\vec y)|^2 \dd y_A
\label{eqLem3}
\end{gather}
for some $\tilde \kappa > 0$ independent of $A$, $\ell$ and $\psi$. Here $\vec y_A$ is short for $\{y_i\}_{i\in A}$. Integrating this over the $\{y_j\}_{j\not\in A}$ and summing over $j$ yields \eqref{th} with the exponent $5/3$ replaced by $1$, and $\kappa = \tilde\kappa q^{2/3}$.

To raise the exponent from $1$ to $5/3$, we partition $B_j$ into $\mu^3$ disjoint cubes $\{C_k\}_k$ of side length $\ell/\mu$ for some integer $\mu\geq 1$. 
We use the identity
\begin{gather}
1 = \sum_{\substack{Q \subseteq A}} \prod_{s \in Q}  \chi_{C_k}(x_s) \prod_{t \in Q^c} \chi_{C_k^c}(x_t)
\label{identSimple} 
\end{gather}
for $\vec x_A\in B_j^{|A|}$, where $Q^c$  denotes $A \setminus Q$ and $C_k^c =B_j\setminus C_k$. By plugging  \eqref{identSimple} into \eqref{eqLem3} we obtain
\begin{gather*}
\sum_{i \in A} \int_{B_j^{|A|}} g(\vec x)^2 |\nabla_i \psi(\vec y)|^2 \dd \vec y_A = \sum_{i \in A} \sum_{k=1}^{\mu^3} \int_{B_j^{|A|}} \chi_{C_k}(x_i) g(\vec x)^2 |\nabla_i \psi(\vec y)|^2 \dd \vec y_A\\
= \sum_{i \in A} \sum_{k=1}^{\mu^3} \sum_{Q \subseteq A}  \int_{B_j^{|A|}} \prod_{s \in Q} \chi_{C_k}(x_s) \prod_{t \in Q^c}  \chi_{C_k^c}(x_t) \chi_{C_k}(x_i) g(\vec x)^2 |\nabla_i \psi(\vec x)|^2 \dd \vec x_A \label{eqC411} \\
 =  \sum_{k=1}^{\mu^3} \sum_{Q \subseteq A} \sum_{i \in Q} \int_{B_j^{|A|}} \prod_{s \in Q}   \chi_{C_k}(x_s) \prod_{t \in Q^c}\chi_{C_k^c}(x_t) g(\vec x)^2 |\nabla_i \psi(\vec x)|^2 \dd \vec x_A \numberthis \label{eqLem4}
 \end{gather*}
For the integration over $\{y_s\}_{s\in Q}$ we can  again use \eqref{eqLem3}, with suitably rescaled variables to replace the integration over $B_j$ with the one over $C_k$. (Note that $g$ is homogeneous of order $-1$ and satisfies the simple scaling property $g(\lambda \vec x) = \lambda^{-1} g(\vec x)$ for $\lambda >0$.)
This yields the bound 
\begin{gather*}
\eqref{eqLem4}  \ge \sum_{k=1}^{\mu^3} \sum_{Q \subseteq A} \frac {\mu^2 \tilde \kappa} {\ell^2} \left( |Q| - q\right)   \int_{C_k^{|Q|}} \dd \vec y_Q \, \int_{{C_k^c}^{(|A|-|Q|)}} \dd \vec y_{Q^c} \, g(\vec x)^2 |\psi(\vec y)|^2 \label{eqC412}  \\
= \frac {\mu^2 \tilde \kappa}{\ell^2}(|A| -  \mu^3 q)  \int_{B_j^{|A|}} g(\vec x)^2 |\psi(\vec y)|^2 \dd \vec y_A \numberthis \label{eqLem5}
\end{gather*}
In the last step, we used again the identity \eqref{identSimple} as well as 
\begin{gather}
|A|  = \sum_{k=1}^{\mu^3}  \sum_{\substack{Q \subseteq A}} |Q| \prod_{s \in Q} \chi_{C_k}(x_s) \prod_{t \in Q^c} \chi_{C_k^c}(x_t)
\end{gather}
Since the left side of \eqref{eqLem5} is obviously non-negative, we can replace $|A| -  \mu^3 q$ by its positive part on the right side.

It remains to choose $\mu$. 
If we ignore the restriction that $\mu\geq 1$ is an integer, we would choose $\mu = (2/5) (|A|/q)^{1/3}$ to obtain the desired coefficient $\propto |A|^{5/3}/q^{2/3}$. It is easy to see that 
\begin{equation}
\sup_{\mu\in\N}  \mu^2 \left[  |A| -  \mu^3 q \right]_+ \geq \frac c {q^{2/3}} \left[ |A| - q\right]_+^{5/3}
\end{equation}
for some universal constant $c>0$. This proves the desired bound, with $\kappa = \tilde \kappa c$.
\end{proof}

\section{A bound on the entropy}\label{sec:entropy}

In this section we shall use the estimates above to give a rough bound on 
\begin{equation}
N_g(E) = \tr \chi_{H_g < E}\,,
\end{equation}
that is, the maximal number of orthonormal  functions in $\mathcal{A}_q^N$ with $\E_g(\psi)< E$, for some (large) $E$. Its logarithm is, by definition, the entropy. Using the localization technique described in Section~\ref{secBoundaryLocalized}, the min-max principle implies that 
\begin{equation}\label{nene}
N_g(E) \leq \sum_{\vn} N_g^{\vn} (E) 
\end{equation}
where $N_g^\vn(E)$ is the maximal number of orthonormal functions in $\mathcal{A}_q^{N,\ell}(\vn)$ with $\Egl(\psi)<E$. Given $E$, we shall choose $\ell$ small enough such $E\ell^2 q^{2/3}\leq \kappa$, with $\kappa$ the constant in Prop.~\ref{prop:nj}. As remarked there, this implies that $n_j \leq q$ 
 for all $1\leq j\leq M$.
 
We will actually show that if $E\ell^2 $ is small enough, then the spectral gap for an excitation is  larger than $E$, and hence $N_g^{\vn} (E)$ is simply equal to the dimension of the space of ground states. 

\begin{lemma}\label{lem:neg}
There exists a universal constant  $c >0$ such that if we choose $E\ell^2  \leq c$, then  
\begin{equation}
N_g^{\vn} (E)  = \prod_{j=1}^M  \binom{q}{n_j}
\end{equation}
\end{lemma}

\begin{proof}
With the aid of  \eqref{lbegl} we have
\begin{equation}\label{insr}
 \Egl(\psi) \geq \left( K_- + \frac V {4\sqrt{3} \ell} \right)^{2}  \E^\ell(\psi)
\end{equation}
for $\psi \in \mathcal{A}_q^{N,\ell}(\vn)$. 
The ground states of the operator corresponding to the quadratic form $\E^\ell$ are all constant, i.e., they are simply products of anti-symmetric functions of the spin variables corresponding to each box, and have zero energy. The spectral gap above the ground state energy  is given by $(\pi/\ell)^2$. With $P_0$ denoting the projection in $L^2({\mathcal{B}}(\vn),\dd\vec y)$ onto the ground state space, we thus have
\begin{equation}\label{spg}
\E^\ell(\psi) \geq \frac{\pi^2}{\ell^2} \left\| (1-P_0) \psi\right\|^2
\end{equation}
In order to bound the norm on the right side from below  in terms of the weighted $\|\,\cdot\,\|_g$ norm, we shall use Lemma~\ref{lemma53}. In \eqref{lbng}, we can simply bound 
\begin{equation}
 \sum_{i=1}^N  \left|\mathcal{N}[i] \right| \int |\nabla_i \psi(\vec y)|^2 g(\vec x)^2\, \dd\vec y  < E  \norm{\psi}_g^2 \sum_{i=1}^N  \left|\mathcal{N}[i] \right|  = E V \norm{\psi}_g^2
 \end{equation}
to obtain 
 \begin{equation}
 \|\psi\|_g^2  \leq    \left[(1+\ve) K_+^2 + \frac {c_1}{\ell^2} \left( 1 + \ve^{-1} \right) V^2  \right] \| \psi\|^2 +  48 c_0  \left( 1 + \ve^{-1}\right)     E \ell^2  \|\psi\|_g^2 
\end{equation}
for any $\ve>0$ and any $\psi\in \mathcal{A}_q^{N,\ell}(\vn)$ with $\Egl(\psi)< E \|\psi\|_g^2$. If $E\ell^2$ is small, we can take $\ve=1$ to  conclude that 
\begin{equation}
 \|\psi\|_g^2  \leq    c  \left[  K_+^2 +  V^2 \ell^{-2}  \right] \| \psi\|^2  
\end{equation}
Moreover, note that $K_+ \leq (1+2\sqrt{3}) K_-$, since $d_{jk}>0$ actually implies $d_{jk}\geq \ell$. We thus also have that  
\begin{equation}
 \|\psi\|_g^2  \leq    c   \left( K_- + \frac V {4\sqrt{3} \ell} \right)^{2}  \| \psi\|^2  
\end{equation}
Applying this to  $(1-P_0)\psi$ in \eqref{spg} and inserting the resulting bound in \eqref{insr} we obtain
\begin{equation}
 \Egl(\psi) \geq c \ell^{-2}  \left\| (1-P_0) \psi\right\|_g^2
\end{equation}
Finally, note that the ground states of $\Egl$ and $\E^\ell$ actually agree, up to a multiplicative normalization constant. Hence, if $\psi$ is orthogonal to a ground state with respect to the inner product $\langle\,\cdot \, | \, \cdot\, \rangle_g$, then 
\begin{equation}
\left\| (1-P_0) \psi\right\|_g^2 = \left\|  \psi\right\|_g^2 + \left\| P_0 \psi\right\|_g^2 \geq \left\|  \psi\right\|_g^2 
\end{equation}
This concludes the proof.
\end{proof}

In combination with \eqref{nene}, Lemma~\ref{lem:neg} yields the bound
\begin{equation}
N_g(E) \leq \sum_{\vn} \prod_{j=1}^M  \binom{q}{n_j} = \binom{qM}{N} \leq \left( \frac { q M e}{N} \right)^N
\end{equation}
for $E\ell^2\leq c$. We recall that the number of boxes is $M=(L/\ell)^3 = N/(\rho\ell^3)$, which is large for $E\ell^2 \sim 1$ and $E\gg  L^{-2}$. Hence we get the upper bound
\begin{equation}\label{bnge}
N_g(E) \leq    \left( c \frac {q  E^{3/2} }{\rho }\right)^N
\end{equation}
for  a suitable constant $c>0$.  This bound readily implies the following proposition.

\begin{proposition}\label{prop:ent}
Let $\{E_j\}_j$ denote the eigenvalues of the Hamiltonian $H_g$ associated to the quadratic form $\E_g$ in \eqref{defeg} on $\mathcal{A}_q^N$. 
For given $\eta = \beta\rho^{2/3}$ there exists a  $c_\eta>0$ such that if $\bar E \geq  c_\eta \beta^{-1} N \ln N$ then 
\begin{gather}
 \sum_{E_j \geq  \bar E } e^{-\beta E_j} \leq  2\, e^{- \frac 12  \beta \bar E }
\end{gather}
\end{proposition}

\begin{proof}
We have
\begin{equation}
 \sum_{E_j \geq  \bar E } e^{-\beta E_j} \leq  \sum_{k\geq 0 } N_g( (k+2 ) \bar E)e^{- (k+1) \beta \bar E} 
\end{equation}
and thus the result follows if 
\begin{equation}
N_g ( (k+2 ) \bar E)e^{- (k+\frac 12) \beta \bar E}  \leq \frac 1{2^k}
\end{equation}
for all $k\geq 0$. Using the bound \eqref{bnge} one easily checks that this is the case under the stated condition on $\bar E$ for suitable $c_\eta$.
\end{proof}

For evaluating the free energy, we can thus limit our attention to eigenvalues $E_j$ satisfying $\beta E_j \leq c_\eta N \ln N$ for suitable $c_\eta>0$. We shall show in the next section that in this low energy sector the eigenvalues are well approximated by the corresponding ones for non-interacting particles.

\section{Comparison with non-interacting particles in the low-energy sector}\label{sec:low}

We shall now investigate the bounds derived in Section~\ref{sec:enno} more closely and apply them to the low energy sector, where 
$\E_g(\psi)\leq E \|\psi\|_g^2$ for some $E\lesssim N \ln N$. 
We again localize the particles into boxes, this time with much larger $\ell$, however. 
We start with the estimate on the ratio of the norm $\|\psi\|_g$ to the  standard, non-weighted $L^2$ norm $\|\psi\|$. 

\begin{proposition}\label{prop:71}
Let $\psi \in \mathcal{A}_q^{N,\ell}(\vn)$ satisfy $\Egl(\psi)\leq E \|\psi\|_g^2$ for some  $E$ with $E\ell^2 \gtrsim 1$ for large $N$. Then
\begin{equation}\label{prop1}
1\geq  \left(  K_- + \frac V {4\sqrt{3} \ell}\right)^2 \frac{ \|\psi\|^2}{ \|\psi\|_g^2} \geq 1-\delta
\end{equation}
with
\begin{equation}\label{def:delta}
\delta \leq c \left[ q^{1/5} (E\ell^2)^{3/10} N^{-1/3} (\rho\ell^3)^{-1/6} +q^{2/5}(E\ell^2)^{11/10} N^{-7/6} (\rho\ell^3)^{-1/3}  \right]
\end{equation}
with $K_-$ and $V$  defined in \eqref{def:kmv}.
\end{proposition}

We note that $\delta$ is small if 
\begin{equation}
E \ell^2 \ll \min\{ N^{10/9} (\rho \ell^3)^{5/9} ,  N^{35/33} (\rho \ell^3)^{10/33} \}
\end{equation}
which gives us freedom to choose $\ell$ large while $E \lesssim N \ln N$. We will choose $\ell \sim N^{\nu}$ for rather small $\nu$ below, in which case the first term in \eqref{def:delta} will be dominating.

\begin{proof}
The first bound in \eqref{prop1} follows immediately \eqref{lbng}. For the lower bound, we use 
\begin{equation}
 \sum_{i=1}^N  \left|\mathcal{N}[i] \right| \int |\nabla_i \psi(\vec y)|^2 g(\vec x)^2\, \dd\vec y \leq 27 \bar n E \|\psi\|_g^2
\end{equation} 
in \eqref{lbng}, 
where we denote $\bar n = \max_{j} n_j$. We can also bound $V\leq 27 \bar n N$ and  
\begin{equation}
K_- + \frac V {4\sqrt{3} \ell} \geq  \frac{ N(N-1)  }{2 \sqrt{3} L }
\end{equation} 
The second bound in \eqref{lbng} thus becomes
\begin{equation}\label{sbtb}
\left[ 1 -\left( 1 + \ve^{-1}\right)   {c_0}  \frac{12 L^2 (27 \bar n)^2 }{ N(N-1)^2 }    E\right]  \|\psi\|_g^2  \leq    \left[(1+\ve) K_+^2 + \frac {c_1}{\ell^2} \left( 1 + \ve^{-1} \right) V^2  \right] \| \psi\|^2  
 \end{equation}
for arbitrary $\ve>0$. By assumption $E \ell^2$ is not small, hence we have $\bar n  \leq c q^{2/5} (E\ell^2)^{3/5}$, as remarked after Proposition~\ref{prop:nj}. 

It remains  to estimate the ratio $K_-/K_+$. We distinguish the contribution to the sum coming from $d_{jk} < r \sqrt{3}\ell$ and $d_{jk}\geq r \sqrt{3}\ell$, respectively, for some large integer $r$ to be chosen below. We have 
\begin{align}\nonumber
K_+ - K_- & = \sum_{\substack{1 \le j < k \le M\\ d_{jk} > 0}} \frac {n_j n_k}{d_{jk}} \frac {2 \sqrt{3} \ell}{d_{jk} + 2 \sqrt{3} \ell} \\ \nonumber &  \le \bar n\sum_{\substack{1 \le j < k \le M\\ 0 < d_{jk} < r \sqrt 3 \ell}} \frac {n_j }{d_{jk}} \frac {2 \sqrt{3} \ell}{d_{jk} + 2 \sqrt{3} \ell} + \lk (1 + \frac r 2 \rk)^{-1} \sum_{\substack{1 \le j < k \le M\\ d_{jk} \ge r \sqrt 3 \ell}} \frac {n_j n_k}{d_{jk}} \\ & 
\le c \frac {\bar n r N} \ell  +  \lk (1 + \frac r 2 \rk)^{-1} K_+
\end{align}
By optimizing over $r$ as well as $\ve$ and using that $\bar n \leq c q^{2/5} (E\ell^2)^{3/5}$ we arrive at the desired result. 
\end{proof}

In combination with \eqref{lbegl}, Proposition~\ref{prop:71}  yields the lower bound
\begin{equation}\label{compes}
\frac{\Egl(\psi)}{\|\psi\|_g^2} \geq \frac{\E^\ell(\psi)}{\|\psi\|^2}\left(1-\delta\right)
\end{equation}
for $\psi\in\mathcal{A}_q^{N,\ell}(\vn)$ in the low energy sector $\Egl(\psi) < E$. This allows us to compare our model directly with non-interacting particles. Note that the eigenfunctions of the operator corresponding to the quadratic form on the right side are tensor products over different boxes and, in particular, the eigenvalues are simply sums over the corresponding eigenvalues of free fermions in each box. 
The bound \eqref{compes} does not directly give us lower bounds on the eigenvalues of $H_g$, except for the lowest one, however. 
To complete the proof, we have to estimate the difference between the inner product $\langle\,\cdot\,|\,\cdot\,\rangle_g$ and the standard inner product on $L^2$, denoted by $\langle\,\cdot\,|\,\cdot\,\rangle$ in the following.

We define the multiplication operator 
\begin{equation}
G  = \left( K_- + \frac V {4\sqrt{3} \ell} \right)^{-1} g(\vec x)
\end{equation}
which is larger or equal to $1$ by \eqref{lbg}. The bound \eqref{lbegl} thus reads
\begin{equation}
\frac{\Egl(\psi)}{\|\psi\|_g^2} \geq \frac{ \E^\ell(\psi)}{\| G \psi\|^2} = \frac { \langle \phi | G^{-1} H G^{-1} | \phi\rangle}{\|\phi\|^2}
\end{equation}
where we introduced $\phi = G \psi$ and denoted by $H$ the Hamiltonian for non-interacting particles, i.e., the Laplacian on ${\mathcal{B}}(\vn)$ with Neumann boundary conditions. Note that the orthogonality condition $\langle \psi_j|\psi_k\rangle_g = 0$ is equivalent to $\langle \phi_j|\phi_k\rangle=0$. 
Given some $E_0>0$, we define the cut-off Hamiltonian
\begin{gather*}
H_c =  H \,\theta( E_0 - H )\,, \numberthis
\end{gather*}
with $\theta$ denoting the Heaviside step function. This is clearly a bounded operator with $\norm{H_c} \le E_0$. Obviously 
\begin{equation}
 \langle \phi | G^{-1} H G^{-1} | \phi\rangle \geq \norm{H_c^{1/2} G^{-1} \phi}^2
\end{equation}
which we further bound as 
\begin{align*}
\norm{H_c^{1/2} G^{-1} \phi }^2  & \ge \lk (\norm{H_c^{1/2} \phi} - \norm{H_c^{1/2} (1 - G^{-1}) \phi} \rk )^2 \\ 
  & \ge \norm{H_c^{1/2} \phi}^2 - 2 \norm{H_c^{1/2} \phi} \norm{H_c^{1/2}} \norm{(1-G^{-1} ) \phi} \\
 & \ge \norm{H_c^{1/2} \phi}^2 - 2  E_0 \norm{(1-G^{-1} ) \phi} \norm{\phi} \numberthis
\end{align*}
Now
\begin{gather*}
\norm{(1-G^{-1} ) \phi}   \leq \norm {(1- G^{-2})^{1/2} \phi} \leq \delta^{1/2} \norm{\phi} \numberthis
\end{gather*}
where we used $G\geq 1$ in the first and Proposition~\ref{prop:71} in the second step. We conclude that
\begin{equation}\label{conss}
\frac{\Egl(\psi)}{\|\psi\|_g^2} \geq  \frac { \langle \phi | H_c - 2 E_0 \delta^{1/2} | \phi\rangle}{\|\phi\|^2}
\end{equation}
under the conditions stated in Proposition~\ref{prop:71}.

\section{Convergence of the free energy}
\label{secFreeEnergyConvergence}

We now have all the necessary tools to complete the proof of Theorem~\ref{thmAll}.
Proposition~\ref{prop:ent} implies that if we choose $\bar E =  c_\eta \beta^{-1} N \ln N$ for a suitable constant $c_\eta>0$, then 
\begin{gather}\label{eq81}
F_g(\beta,N,L) \geq - T  \ln \left(  2\, e^{- \frac 12  \beta \bar E }  +  \sup_{\substack{\{\psi_k\in \mathcal{A}_q^{N}\} \\ \braket{\psi_i}{\psi_j}_g = \delta_{ij} }} \sum_{k=1}^{N_g(\bar E)}  e^{- \beta \E_g(\psi_k)}
 \right)
\end{gather}
Here $N_g(\bar E)$ denotes the number of states with energy below $\bar E$, which was estimated in \eqref{bnge}. We can write, alternatively, 
\begin{equation}\label{eq82}
\sup_{\substack{\{\psi_k\} \\ \braket{\psi_i}{\psi_j}_g = \delta_{ij} }} \sum_{k=1}^{N_g(\bar E)}  e^{- \beta \E_g(\psi_k)} = \sup_{\substack{\{\psi_k\} ,\, \E_g(\psi_k)< \bar E \\ \braket{\psi_i}{\psi_j}_g = \delta_{ij} }} \sum_{k}  e^{- \beta \E_g(\psi_k)}
\end{equation}
By localizing into small boxes of side length $\ell$ with Neumann boundary conditions, as detailed in Section~\ref{secBoundaryLocalized}, we further have by the min-max principle 
\begin{equation}\label{eq83}
\eqref{eq82} \leq \sum_{\vn} \sup_{\substack{\{\psi \in \mathcal{A}_q^{N,\ell}(\vn)\} ,\, \Egl(\psi)\leq \bar E \\ \braket{\psi_i}{\psi_j}_g = \delta_{ij} }} \sum_{k}  e^{- \beta \Egl(\psi_k)}
\end{equation}
If we choose $\bar E \ell^2 \gtrsim 1$, we can apply the bound \eqref{conss} from the previous subsection. It implies
\begin{equation}\label{eq84}
\eqref{eq83} \leq e^{2\beta E_0 \delta^{1/2}} \sum_{\vn} \sup_{\substack{\{\phi \in G \mathcal{A}_q^{N,\ell}(\vn)\} ,\, \langle\phi_k| H_c |\phi_k\rangle \leq \bar E+ 2 E_0 \delta^{1/2} \\ \braket{\phi_i}{\phi_j} = \delta_{ij} }} \sum_{k}  e^{- \beta  \langle\phi_k| H_c |\phi_k\rangle}
\end{equation}
with $\delta$ defined in Proposition \ref{prop:71}. If we choose $E_0$ such that $\bar E + 2 E_0 \delta^{1/2} \leq E_0$, which is possible for $\delta<1/4$, we can drop the cutoff in $H_c$ and replace $H_c$ by $H$, the Laplacian on $(\bigcup_j B_j)^N$ with Neumann boundary conditions. To obtain an upper bound on \eqref{eq84}, we can then further neglect the  bound on $\langle\phi_k| H |\phi_k\rangle$, and sum over all eigenvalues. We obtain
\begin{equation}\label{eq85}
\eqref{eq84} \leq e^{2\beta E_0 \delta^{1/2}} e^{- \beta F(\beta,N,L,\ell)}
\end{equation}
where $F(\beta,N,L,\ell)$ denotes the free energy of non-interacting fermions in $\bigcup_j B_j$ (with Neumann boundary conditions on the boundaries of the $B_j$). In particular, in combination \eqref{eq81}--\eqref{eq85} imply 
\begin{equation}\label{eq86}
F_g(\beta,N,L) \geq F(\beta,N,L,\ell) - 2 E_0 \delta^{1/2} - T \ln \left( 1 + 2 \, e^{- \frac 12  \beta \bar E }   e^{-2\beta E_0 \delta^{1/2}} e^{ \beta F(\beta,N,L,\ell)} \right) 
\end{equation}

We will choose $\ell\gtrsim 1$, in which case $F(\beta,N,L,\ell) \sim N$ and hence the last term in \eqref{eq86} is, in fact, exponentially small in $N$,  since $\bar E \sim N \ln N$. To complete the proof, it suffices to observe that 
\begin{equation}
F(\beta,N,L,\ell)  \geq F(\beta,N,L) - c_\eta \frac {N \rho^{1/3}} {\ell} 
\end{equation}
which is an easy exercise. 
To minimize the total error, we shall choose 
\begin{equation}
\ell \sim \rho^{-1/3} N^{1/63} \left( \ln N \right)^{-23/21}
\end{equation}
to obtain
\begin{equation}
F_g(\beta,N,L)  \geq F(\beta,N,L) - c_\eta \rho^{2/3} N^{62/63} \left( \ln N \right)^{23/21} 
\end{equation}
This completes the proof of Theorem~\ref{thmAll}.

\section*{Acknowledgments}
Financial support by the Austrian Science Fund (FWF), project Nr. P 27533-N27, is gratefully acknowledged. 

\makeatletter
\renewcommand\@biblabel[1]{#1.}
\makeatother

\end{document}